\documentclass[a4paper,11pt]{article}
\usepackage{fullpage}
\usepackage{lmodern}
\usepackage{xspace}
\usepackage{amssymb}
\usepackage{amsmath}
\usepackage{amsthm}
\usepackage{thmtools,thm-restate}
\usepackage{graphicx}
\newenvironment{di}{\list{$\bullet$}{\itemsep0pt\parsep0pt}}{\endlist}
\newenvironment{claimproof}{\begin{proof}}{\end{proof}}
\let\subparagraph\paragraph

\usepackage[T1]{fontenc}

\usepackage{microtype}
\usepackage{mathtools}

\usepackage{hyperref}

\hypersetup{%
  breaklinks,%
  ocgcolorlinks,
  colorlinks=true,%
  linkcolor=[rgb]{0.45,0.0,0.0},%
  citecolor=[rgb]{0,0,0.45}
}

\newtheorem{theorem}{Theorem}
\newtheorem{lemma}[theorem]{Lemma}
\newtheorem{claim}[theorem]{Claim}
\newtheorem{corollary}[theorem]{Corollary}

\let\geq\geqslant
\let\leq\leqslant
\let\ge\geqslant
\let\le\leqslant

\newcommand{\opmaketree}{\mathsf{maketree}}
\newcommand{\oplink}{\mathsf{link}}
\newcommand{\opcut}{\mathsf{cut}}
\newcommand{\opfindroot}{\mathsf{findroot}}
\newcommand{\opevert}{\mathsf{evert}}
\newcommand{\opexpose}{\mathsf{expose}}
\newcommand{\opfindcost}{\mathsf{findcost}}
\newcommand{\opaddcost}{\mathsf{addcost}}
\newcommand{\opfindblue}{\mathsf{findblue}}
\newcommand{\opaddblue}{\mathsf{addblue}}
\newcommand{\opaddred}{\mathsf{addred}}

\newcommand{\bitreverse}{\mathit{reverse}}
\newcommand{\bitparent}{\mathit{parent}}
\newcommand{\bitblue}{\mathit{blue}}
\newcommand{\bitred}{\mathit{red}}
\newcommand{\bitminblue}{\mathit{minblue}}
\newcommand{\bitminred}{\mathit{minred}}
\newcommand{\bitDblue}{\Delta\mathit{blue}}
\newcommand{\bitDred}{\Delta\mathit{red}}
\newcommand{\bitDminblue}{\Delta\mathit{minblue}}
\newcommand{\bitDminred}{\Delta\mathit{minred}}
\newcommand{\Reals}{\mathbb{R}}

\newcommand{\eps}{\varepsilon}  

\DeclareMathOperator{\polylog}{polylog}
\DeclareMathOperator{\Dgm}{Dgm}

\makeatletter
\partopsep\z@ \textfloatsep 10pt plus 1pt minus 4pt
\def\section{\@startsection {section}{1}{\z@}%
  {-3.5ex plus -1ex
    minus -.2ex}{2.3ex plus .2ex}{\large\bf}}
\def\subsection{\@startsection{subsection}{2}%
  {\z@}{-3.25ex plus
    -1ex minus -.2ex}{1.5ex plus .2ex}{\normalsize\bf}}
\def\@fnsymbol#1{\ensuremath{\ifcase#1\or 1\or 2\or 3\or 4\or
    5\or 6\or 7 \or 8\ or 9 \or 10\or 11 \else\@ctrerr\fi}}
\makeatother

\title{Geometric Matching and Bottleneck Problems}

\author{Sergio Cabello%
  \thanks{Faculty of Mathematics and Physics, University of Ljubljana, Slovenia, 
	and Institute of Mathematics, Physics and Mechanics, Slovenia. 
	sergio.cabello@fmf.uni-lj.si}
  \and
  Siu-Wing Cheng%
  \thanks{Department of Computer Science and Engineering,
    HKUST, Clear Water Bay, Hong Kong. scheng@cse.ust.hk}
  \and
  Otfried Cheong%
  \thanks{SCALGO, Aarhus, Denmark. otfried@scalgo.com}
  \and
  Christian Knauer%
  \thanks{Department of Computer Science, University of Bayreuth,
    Germany. christian.knauer@uni-bayreuth.de}}

\newcommand{\OurFunding}{This research was funded in part by the Slovenian Research and
Innovation Agency (P1-0297, J1-2452, N1-0218, N1-0285), in part by the
Research Grants Council, Hong Kong, China (project no.~16207419), and
in part by the European Union (ERC, KARST, 101071836).  Views and
opinions expressed are however those of the authors only and do not
necessarily reflect those of the European Union or the European
Research Council.  Neither the European Union nor the granting
authority can be held responsible for them.}

\newcommand{\OurAcknowledgements}{The authors thank Sang Won Bae for helpful discussions.}

\begin{document}

\maketitle

\begin{abstract}
  Let $P$ be a set of at most~$n$ points and let $R$ be a set of at
  most~$n$ geometric ranges, such as for example disks or rectangles,
  where each~$p \in P$ has an associated \emph{supply} $s_{p} > 0$,
  and each~$r \in R$ has an associated \emph{demand}~$d_{r} > 0$.  A
  (many-to-many) \emph{matching} is a set $\mathcal{A}$ of ordered
  triples~$(p,r,a_{pr}) \in P \times R \times \Reals_{>0}$ such
  that~$p \in r$ and the $a_{pr}$'s satisfy the constraints given by
  the supplies and demands.  We show how to compute a maximum
  matching, that is, a matching maximizing~$\sum_{(p,r,a_{pr}) \in
    \mathcal{A}} a_{pr}$.

  Using our techniques, we can also solve minimum bottleneck problems,
  such as computing a perfect matching between a set of~$n$ red
  points~$P$ and a set of~$n$ blue points~$Q$ that minimizes the
  length of the longest edge. For the~$L_\infty$-metric, we can do
  this in time~$O(n^{1+\eps})$ in any fixed dimension, for
  the~$L_2$-metric in the plane in time~$O(n^{4/3 + \eps})$, for
  any~$\eps > 0$.
\end{abstract}

\section{Introduction}

The classic \emph{assignment} problem is to assign $n$~items (such as
jobs) to $n$~other items (such as machines).  It asks for a bijection
between the two sets that maximizes some quality or minimizes some
cost function.  Computing a perfect matching in a bipartite graph is
one special case of the assignment problem.  The assignment problem
stood at the cradle of the field of combinatorial optimization in the
1950's, and influenced the development of integer and linear
programming and the theory of network flows.  The book by Burkard et
al.~\cite{burkard} presents the history of the problem and results for
different cost functions.

A natural generalization is to relax the one-to-one requirement.
Problems of this nature arise commonly in practice, for instance when
assigning mobile clients to base stations, patients to hospitals,
dinner guests to restaurants, etc.  Such problems are commonly solved
using maximum-flow, see for instance Exercises~7, 8, 9, 16, 19, 26 in
Chapter~7 of Kleinberg and Tardos' book~\cite{kleinberg-tardos}.

\subparagraph*{Geometric matching problems.}

In several settings one encounters assignment or matching problems
between geometric objects. Let us define these geometric matching
problems more precisely.  Let $P$ be a set of points and let $R$ be a
set of geometric ranges, such as for example disks or rectangles. We
use $n$ to denote $|P|+|R|$.  Their \emph{incidence graph}, denoted by
$I(P,R)$, is the bipartite graph (see Figure~\ref{fig:incidences} for
an example)
\[ 
I(P,R) ~=~ \big(P\sqcup R, \{pr \mid p\in P, r\in R, p\in r \} \big).
\]

\begin{figure}
  \centerline{\includegraphics[width=.8\textwidth,page=3]{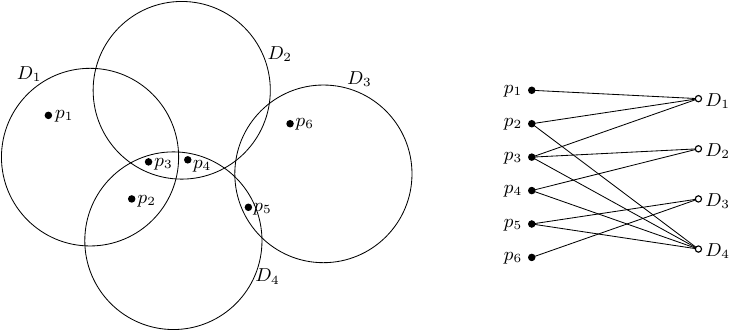}}
  \caption{A set of six points and a set of four congruent disks in the plane (left) 
    and their incidence graph (right).}
  \label{fig:incidences}
\end{figure}

Each point~$p \in P$ has an associated \emph{supply}~$s_{p} > 0$ and
each range~$r \in R$ has an associated \emph{demand}~$d_{r} > 0$.  We
consider the elements of~$P$ as \emph{providers} and the elements
of~$R$ as \emph{consumers} of a common commodity that is infinitely
divisible.  The commodity can be sent only along the edges of~$I(P,R)$,
and the goal is to maximize the amount of commodity that is sent while
respecting the supplies and the demands.  The \emph{target value} of
the instance is
\[
\mu ~=~ \min\bigg\{ \sum_{p\in P} s_p,\; \sum_{r\in R} d_r \bigg\},
\]
which is an obvious upper bound on the amount of the commodity that can be sent.

Formally, a (many-to-many) \emph{matching} is a set $\mathcal{A}$ of
ordered triples~$(p,r,a_{pr})\in P\times R\times\Reals_{>0}$ such that
\begin{di}
\item for every $(p,r,a_{pr}) \in \mathcal{A}$, $pr$ is an edge of $I(P,R)$;
\item for every $(p,r) \in P\times R$, there is at most one triple $(p,r,\cdot)$ in $\mathcal{A}$;
\item for every point $p \in P$, $\sum_{r\in R:~(p,r,a_{pr}) \in {\cal A}} a_{pr} \leq s_p$;
\item for every range $r \in R$, $\sum_{p\in P:~(p,r,a_{pr}) \in {\cal A}} a_{pr} \leq d_r$.
\end{di}
The \emph{value} of the matching $\mathcal{A}$ is
$\eta=\eta(\mathcal{A})=\sum_{(p,r,a_{pr}) \in {\cal A}} a_{pr}$, and
the objective is to maximize this value.  A matching is
\emph{satisfying} when its value is the target value, that is,
when~$\eta=\mu$.  The \emph{size} of a matching is $|\mathcal{A}|$,
that is, the number of edges of $I(P,R)$ that carry part of the
commodity.

A geometric matching problem can be solved using a standard
transformation to a maximum-flow instance built around~$I(P,R)$ (with
directed edges).  However, any approach that explicitly constructs
$I(P,R)$ is going to spend time at least proportional to its size,
which may be $\Theta(n^2)$. Thus, even when using the recent
breakthrough near-linear time algorithm of Chen et
al.~\cite{ChenKLPGS22} for maximum-flow, we will spend quadratic time.

A key question is whether one can use the geometry of the input to
solve the matching problem without explicitly constructing~$I(P,R)$.
Our contribution is to note that a compact representation of the
incidences of geometric objects can indeed be used to solve the
geometric matching problem without constructing~$I(P,R)$. More
precisely, we use a compact representation of $I(P,R)$ as the union of
complete bipartite graphs.  We denote by $\sigma$ the size of this
compact representation, which will be defined in
Section~\ref{sec:compact}.  We provide two algorithms for the maximum
matching problem depending on what we want to assume about the
supplies and the demands.

\subparagraph*{Integral supplies and demands.}

Our first algorithm uses the recent algorithm of Chen et
al.~\cite{ChenKLPGS22} to solve maximum-flow instances in near-linear
time and assumes that the supplies and the demands are integers.
Recall that $\mu = \min\bigl\{\sum_{p \in P} s_p, \sum_{r \in R}
d_r\bigr\}$.  The running time is $O(\sigma^{1+\eps} \log \mu)$ with
high probability (more specifically, the probability is at least $1 -
\sigma^{-O(1)}$), for any constant~$\eps>0$.  Retrospectively, this
result is a simple combination of the recent breakthrough in computing
maximum flow and compact representations of incidences in geometric
settings, a common tool in computational geometry.

This first algorithm has interesting consequences for several
fundamental problems concerning one-to-one matchings in geometric
settings. However, not every solution to the matching problem with
unit supplies and demands is a one-to-one matching, because the
solution may use fractional flow values on some edges.  For example,
on a regular bipartite graph solving the matching problem with unit
supplies and demands is trivial using fractional flows, but it is
still a challenge to determine an integral maximum flow.  We solve
this problem using the approach in~\cite{GoelKK13,Madry13}.

Here is a list of consequences of our new algorithm; in all cases, $n$
is the number of points and ranges, the results hold with high
probability, and $\eps$ is an arbitrary positive constant that affects
the constants hidden in the $O$-notation:
\begin{itemize}
\item When $R$ is a set of axis parallel rectangles in the plane, we
  can compute a maximum one-to-one matching in~$I(P,R)$ in
  $O(n^{1+\eps})$ time.
\item When $R$ is a set of congruent disks in the plane, we can
  compute a maximum one-to-one matching in~$I(P,R)$ in
  $O(n^{4/3+\eps})$ time.
\item For two given point sets $P$ and~$Q$ in the plane, and a real
  value~$\lambda\ge 0$, we can compute a maximum one-to-one matching
  in the bipartite graph on~$P\sqcup Q$ that connects points at
  distance at most~$\lambda$. In the $L_\infty$-metric, the running
  time is~$O(n^{1+\eps})$, while in the $L_2$-metric the running time
  is~$O(n^{4/3+\eps})$.
\item We can decide in $O(n^{1+\eps})$ time if the bottleneck distance
  of two persistence diagrams with $n$~points outside the diagonal is
  smaller or equal to a given value~$\lambda\ge 0$.
\end{itemize}
For all problems, the previous best algorithm had a running time of
roughly $O(n^{3/2})$ and was due to Efrat, Itai and
Katz~\cite{efrat}. That result was an efficient version of the
maximum-matching algorithm by Hopcroft and Karp~\cite{hopcroft-karp},
where one makes $O(n^{1/2})$ rounds, and in each round one finds
several augmenting paths of the same length. Using range searching
data structures, each round is implemented in near-linear time.

If one is willing to use as a parameter the so-called density $\rho$
of the geometric objects, then Bonnet, Cabello, and
Mulzer~\cite{BonnetCM} provided an algorithm to compute a maximum
(one-to-one) matching of $I(P,R)$ in $O(\rho^{3\omega/2}
n^{\omega/2})$ time with high probability, where $\omega>2$ is a
constant such that any two $n\times n$ matrices can be multiplied in
time~$O(n^\omega)$. With the current bounds on~$\omega$, the running
time is~$O(\rho^{3.56}n^{1.19})$. Note that in the worst case
$\rho=\Theta(n)$ and then the algorithm takes $\Omega(n^4)$ time, even
for~$\omega$ arbitrarily close to~$2$.  For congruent disks and the
current value of~$\omega$, our new bound is better when
$\rho=\Omega(n^{0.041})$; if~$\omega$ can be made arbitrarily close
to~$2$, our new bound for congruent disks is still better
when~$\rho=\Omega(n^{1/9+\eps})$ for some~$\eps>0$.  For axis-parallel
rectangles and the current value of~$\omega$, our new bound is better
for any density~$\rho$; if~$\omega$ can be made arbitrarily close
to~$2$, our new bound for rectangles is better as soon
as~$\rho=\Omega(n^{\eps})$ for some $\eps>0$.  They give other results
that are independent of the density~$\rho$, but those do not apply to
the bipartite setting that we consider here.

Har-Peled and Yang~\cite{Har-PeledY22} considered
obtaining $(1+\eps)$-approximations to the maximum matching in
near-linear time.

\subparagraph*{Arbitrary supplies and demands.}

In our second algorithm we do not make any assumption about the
supplies and the demands.  In computational geometry it is common to
allow arbitrary real numbers in the input, because the input may be
coming from a previous computation and may be, say, a Euclidean
distance or the area of a disk.

Our algorithm is a variant of Dinitz' algorithm for maximum-flow as
implemented by Sleator and Tarjan~\cite{tarjan}.  We prune the current
matching after each phase of Dinitz' algorithm, ensuring that the
current flow is non-zero on a linear number of edges only.  We also
store the edges from providers to consumers in each residual graph
implicitly using a clique cover of the graph.  Assuming a compact
representation of the incidence graph~$I(P,R)$ of size~$\sigma$ we
compute a maximum matching in~$O(n \sigma \log \sigma)$ time.

Our time bound is better than using the usual transformation into a
maximum flow and then using the fastest known combinatorial algorithms
for maximum flow that work for arbitrary edge capacities: For
maximum-flow instances on graphs with $|V|$~vertices and $|E|$~edges,
the best known running time is $O(|V||E|)$~\cite{KingRT94,DBLP:conf/stoc/Orlin13,OG2021}.
In our setting, this leads to
a worst-case running time of~$O(n^{3})$ (for the decision problem),
because $I(P,R)$ could be dense.

Our time bound is also better than using the construction used for the
integral case; see Figure~\ref{fig:example2}. In this construction, we
have a graph with $O(\sigma)$~edges and vertices, which means that a
maximum flow can be computed in~$O(\sigma^2)$ time. This time bound is
strictly better than $O(n \sigma \log \sigma)$ only when
$n=\omega(\sigma/\log \sigma)$, but this never happens for any of the
geometric scenarios we consider.  Even for points and intervals on the
line we have $\sigma=\Theta(n \log n)$.

Our algorithm has the following consequences for the geometric
matching problem with arbitrary supplies and demands: 
\begin{itemize}
\item When $R$ is a set of axis parallel rectangles in the plane, 
  we can compute a maximum matching in $I(P,R)$ in $O(n^2\log^3 n)$ time. 
\item When $R$ is a set of congruent disks in the plane,
  we can compute a maximum matching in $I(P,R)$ in $O(n^{7/3}\log^2 n)$ time.
\item For two given point sets $P$ and $Q$ in the plane, and a real
  value $\lambda\ge 0$, we can compute a maximum matching in the
  bipartite graph on $P\sqcup Q$ that connects points at distance at
  most~$\lambda$. In the $L_\infty$-metric, the running time
  is~$O(n^2\log^3 n)$, while in the $L_2$-metric the running time
  is~$O(n^{7/3}\log^2 n)$.
\end{itemize}

\subparagraph*{Bottleneck instances.}

Consider the scenario where we have two sets of $n$ points each in the
plane and we want to compute a perfect matching between them that
minimizes the length of the longest edge. That is called a
\emph{bottleneck matching}, and its computation has been considered
before by Efrat, Itai and Katz~\cite{efrat}.  A precise definition is
given in Section~\ref{sec:bottleneck}.

A similar concept of using a perfect matching to define a bottleneck
distance has been introduced in the context of persistence diagrams;
see~\cite{Cohen-SteinerEH07,DW,EH,KerberMN17} or the definition in
Section~\ref{sec:persistence}. This concept is closely related to the
$L_\infty$-metric in the plane.

Similarly, for non-unit demands and supplies we can consider a
matching problem between two sets of points where we want to compute a
satisfying matching that minimizes the length of the longest edge
being used. This is the bottleneck value for the geometric matching
problem.  Long and Wong~\cite{long2017} introduced this~\emph{spatial
matching problem}, which they denoted~\textsc{spm-mm}.  Closely
related problems have been considered in the field of spatial
databases, for example when assigning mobile clients to base stations
with limited capacity~\cite{long2013, leong}.  In this context, one
uses the $L_2$-metric.

In all these problems, we can solve the \emph{decision problem} using
the algorithms described above: For a given~$\lambda$, is there a
perfect matching or a satisfying matching where all edges have
length at most~$\lambda$?  We can then perform a binary search in a
set of (so-called \emph{critical}) values that is known to contain the
optimal value, without substantially increasing the running time.  We
obtain the following results, where, as before, the running times hold
with high probability and $\eps$ is an arbitrarily small positive
constant:
\begin{itemize}
\item The bottleneck matching for two sets of $n$~points in the plane
  under the $L_2$-metric can be computed in in~$O(n^{4/3+\eps})$ time.
  The previous time bound was $O(n^{3/2}\log n)$~\cite{efrat}.
\item The bottleneck matching for two sets of $n$~points in the plane
  under the $L_1$- or~$L_\infty$-metric can be computed in
  $O(n^{1+\eps})$~time.  The previous bound was~$O(n^{3/2}\log
  n)$~\cite{efrat}.
\item The bottleneck distance of two persistence diagrams with
  $n$~points outside the diagonal can be computed in~$O(n^{1+\eps})$
  time.  The previous bound was~$O(n^{3/2}\log n)$.
\item The bottleneck value for a geometric matching problem can be
  computed in $O(n^{4/3+\eps})$~time, if the supplies and the demands
  are small integers, and in $O(n^{7/3}\log^3 n)$ time for arbitrary
  supplies and demands.
\end{itemize}

Note that the recent result of Katz and Sharir~\cite{KatzS23} is for
the bottleneck matching problem in the \emph{non-bipartite} case; 
the techniques do not apply to the bipartite case we consider here.

It is perhaps surprising that minimum bottleneck matchings or
matchings are not necessarily planar.  Finding a minimum \emph{planar}
bottleneck matching is hard~\cite{carlsson-2015,KARIMABUAFFASH2014447}.

\subparagraph*{Roadmap.}

In Section~\ref{sec:compact} we describe the compact representation of
incidence graphs using complete bipartite cliques.  In
Section~\ref{sec:integral} we present the algorithms for the case of
integral supplies and demands, while in in Section~\ref{sec:arbitrary}
we consider the setting with arbitrary values.  In
Section~\ref{sec:red-blue-link-cut} we present a dynamic tree, a data
structure used in Section~\ref{sec:arbitrary} as a black-box.  We
conclude in Section~\ref{sec:conclusions}.


\section{Compact representation of incidences}
\label{sec:compact}

For any two disjoint sets $U$ and $V$, let $K(U,V)$ denote the complete
bipartite graph with vertex set $U\sqcup V$ and edge set $\{ uv\mid u\in U, v\in V\}$.

Let $G=(U,V,E)$ be a bipartite graph with vertex set~$U \sqcup V$. A
\emph{compact representation} of~$G$ is a collection~$\{ (U_i,V_i)\mid
i\in I\}$ such that
\begin{di}
\item for each $i\in I$, $U_i\subseteq U$ and $V_i\subseteq V$, and
\item $G$ is the union of the complete bipartite graphs $K(U_i,V_i)$ over $i\in I$.
\end{di}
The compact representation is edge-disjoint if the graphs $K(U_i,V_i)$, $i\in I$,
are pairwise edge-disjoint.
The \emph{size} of the compact representation is defined as 
$\sum_{i\in I} (|U_i|+|V_i|)$. The size represents the length of the description
of the representation, assuming that vertices can be represented in constant space
(the bit complexity would be slightly larger).
See Figure~\ref{fig:example1} for an example.

\begin{figure}
  \centerline{\includegraphics[width=\textwidth,page=1]{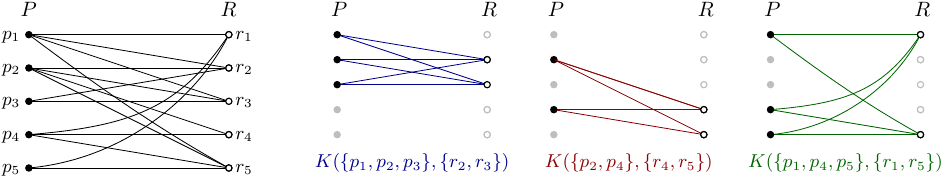}}
  \caption{An example of a bipartite graph (left) represented as the
    union of three complete bipartite graphs (right). In this example,
    the compact representation has size $(3+2)+(2+2)+(3+2)=14$.  Note
    that this compact representation is not edge-disjoint because of
    the edge $p_4r_5$.}
  \label{fig:example1}
\end{figure}

Each bipartite graph $G=(U,V,E)$ has a trivial compact representation
of size $\Theta(|U|+|V|+|E|)$, namely $\{ (\{u\},\{v\})\mid uv\in E \}$.
However, some bipartite graphs have compact representations of size $o(|E|)$.
We next describe a few such cases arising in geometry.

\begin{theorem}[Implicit in~{\cite[Section 3]{ae-survey}} or {\cite[Chapter 5]{bkos-08}}] 
  \label{thm:point+rectangle}
  Let $P$ be a set of at most $n$ points and let $R$ be a set of at
  most $n$ axis-parallel boxes in~$\Reals^d$, for constant~$d$.  We
  can compute an edge-disjoint compact representation of the incidence
  graph~$I(P,R)$ of size $O(n \log^d n)$ in time $O(n\log^d n)$.
\end{theorem}

\begin{theorem}[\cite{katz-1997}]
  \label{thm:point+disk1}
  Let $P$ be a set of at most $n$ points and let $D$ be a set of at
  most $n$ congruent disks in the plane.  We can compute an
  edge-disjoint compact representation of the incidence
  graph~$I(P,D)$ of size $O(n^{4/3} \log n)$ in time $O(n^{4/3} \log
  n)$.
\end{theorem}

\begin{theorem}[Implicit in~\cite{Matousek93}; see also~{\cite[Section 4.3]{ae-survey}}]
  \label{thm:point+halfspace}
  Let $P$ be a set of at most $n$ points and let $H$ be a set of at
  most $n$ halfspaces in~$\Reals^d$, for constant~$d$.  We can compute
  an edge-disjoint compact representation of the incidence
  graph~$I(P,H)$ of size $O(n^{\frac{2d}{d+1}} \polylog n)$ in time
  $O(n^{\frac{2d}{d+1}} \polylog n)$.
\end{theorem}

\subparagraph*{Remark:}
Points and balls of different sizes in $\Reals^d$ can be treated as
points and halfspaces in $\Reals^{d+1}$ using standard techniques.
For example, for~$d=2$, the bound on the size of the compact
representation and the running time to construct it is~$O(n^{3/2}
\polylog n)$.


\section{Geometric matching problems with integral supplies and demands}
\label{sec:integral}

Consider a geometric matching problem for a set $P$ of points and a
set $R$ of geometric ranges, with $n$ objects in total.  Each point
$p\in P$ has a supply $s_p>0$ and each range $r\in R$ has a demand
$d_r>0$.  

Assume that we have a compact representation $\{ (P_i,R_i)\mid i\in
I\}$ of size~$\sigma$ for the incidence graph~$I(P,R)$.  We are going
to reduce the problem of deciding whether there is a satisfying
matching to a max-flow problem in a graph whose size is roughly the
size of the compact representation.

The max-flow instance is constructed as follows (see
Figure~\ref{fig:example2} for an example): Initially we have the
vertices $P\sqcup R$ and no edges.  We then add the following:
\begin{di}
\item We add a new vertex $s$ and, for each $p\in P$, 
  we add the directed edge $(s,p)$ with capacity~$s_p$.
\item We add a new vertex $t$ and, for each range $r\in R$, 
  we add the directed edge $(r,t)$ with capacity~$d_r$.
\item For each $i\in I$, we add a new vertex~$v_i$, 
  directed edges $(p,v_i)$ for all $p\in P_i$,
  and directed edges $(v_i,r)$ for all~$r\in R_i$. 
  These edges have capacity~$\mu = \min\bigl\{\sum_{p \in P} s_p, \sum_{r \in R} d_r \bigr\}$.
  We refer to the set of edges added in this step as~$E_i$.
\end{di}
This finishes the description of the $s$-$t$ max-flow instance, which
we denote by~$G$.  This instance depends on~$P$, $R$, the demands, the
supplies, and the compact representation of~$I(P,R)$ being used.

\begin{figure}
  \centerline{\includegraphics[scale=1.1,page=2]{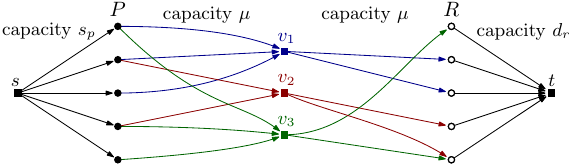}}
  \caption{The max-flow instance $G$ for the compact representation of Figure~\ref{fig:example1}.}
  \label{fig:example2}
\end{figure}

The instance we constructed has $2+|P|+|R|+|I|= 2+n+|I|=O(\sigma)$
vertices and $|P|+|R|+\sum_{i\in I} (|P_i|+|R_i|)=O(\sigma)$ edges.
The time to construct the instance is also $O(\sigma)$, assuming that
the compact representation is already available.  The maximum capacity
in the instance is~$\mu$.

\begin{restatable}{lemma}{flowtomatching}
  \label{le:reduction}
  There is an $s$-$t$ flow of value~$\eta$ in $G$ if and only if there
  exists a matching of value~$\eta$. Moreover, from a flow of
  value~$\eta$ we can recover a matching of value~$\eta$ and
  size~$O(\sigma)$ in~$O(\sigma)$ time.
\end{restatable}
\begin{proof}
  Assume first that there is a matching ${\cal A}$ of value $\eta$.
  We define a flow $f\colon E(G)\rightarrow \Reals_{\ge 0}$ as
  follows.  Initially, we set the flow $0$ at all edges.  For each
  $(p,r,a_{pr})\in {\cal A}$, we select an index $i\in I$ such that
  $(p,r)\in P_i \times R_i$. Such an index $i$ exists because $\{
  (P_i,R_i)\mid i\in I\}$ is a compact representation of the incidence
  graph $I(P,R)$.  We then increase the flow value at each of the
  edges $(s,p)$, $(p,v_i)$, $(v_i,r)$ and $(r,t)$ by~$a_{pr}$.  These
  edges form a path from $s$ to $t$.  It is easy to see that at each
  vertex of~$G$, different from~$s$ and~$t$, we have conservation of
  flow and the capacities of the network are respected.  Therefore we
  have defined a valid flow $f$ in $G$ whose value is
  \[
  \sum_{p\in P} f(s,p) ~=~ 
  \sum_{p\in P} \bigg( \sum_{r\in R: ~(p,r,a_{pr})\in {\cal A}} a_{pr}\bigg) ~=~ 
  \sum_{(p,r,a_{pr})\in {\cal A}} a_{pr} ~=~ 
  \eta.
  \]
  
  Assume now that we have an $s$-$t$ flow $f\colon E(G)\rightarrow
  \Reals_{\ge 0}$ in $G$ with value $\eta$. This means that
  $\eta=\sum_{p\in P} f(s,p)= \sum_{r\in R} f(r,t)$.  For each $i\in
  I$, let $\eta_i$ be the flow through $v_i$, which is $\sum_{p\in P_i}
  f(p,v_i)= \sum_{r\in R_i} f(v_i,r)$.  Since the edges entering into
  the vertex set~$\{ v_i\mid i\in I\}$ form an $s$-$t$ cut, and the
  edges leaving that vertex set also form an $s$-$t$ cut, we have
  \[
  \eta~=~ \sum_{i\in I} \sum_{p\in P_i} f(p,v_i) ~=~ \sum_{i\in I} \sum_{r\in R_i} f(v_i,r) 
  ~=~ \sum_{i\in I} \eta_i.
  \]
  
  Fix any $i\in I$. We construct a matching ${\cal A}_i$ as follows.
  We initialize values $\ell(p)= f(p,v_i)$ for each $p\in P_i$ and
  $\ell(r)= f(v_i,r)$ for each $r\in R_i$, set $\tilde\eta_i=\eta_i$
  and empty matching ${\cal A}_i=\emptyset$.  While $\tilde\eta_i> 0$,
  we select one $p\in P_i$ with $\ell(p)>0$ and one $r\in R_i$ with
  $\ell(r)>0$. Let $\delta=\min\{ \ell(p), \ell(r)\}$.  We then add
  $(p,r,\delta)$ to ${\cal A}_i$, and decrease the values $\ell(p),
  \ell(r), \tilde\eta_i$ by $\delta$.  Throughout the process we
  maintain the invariants
  \begin{align*}
    \ell(u) ~&\ge~ 0 ~~~~\forall u\in P_i\cup R_i,\\
    \eta_i ~&=~ \tilde \eta_i + \sum_{(p,r,a_{pr})\in {\cal A}_i} a_{pr},\\
    \tilde\eta_i ~&=~ \sum_{p\in P_i} \ell(p) ~=~ \sum_{r\in R_i} \ell(r),\\
    \forall p\in P_i:~~~~f(p,v_i) ~&=~ \ell(p) + \sum_{r\in R_i:~ (p,r,a_{pr})\in {\cal A}_i} a_{pr},\\
    \forall r\in R_i:~~~~f(v_i,r) ~&=~ \ell(r) + \sum_{p\in P_i:~ (p,r,a_{pr})\in {\cal A}_i} a_{pr}.
  \end{align*}
  Therefore, when we finish, $\tilde \eta_i=0$ and $\ell(u)=0$ for each
  $u\in P_i\cup R_i$.  We then have at the end
  \begin{align*}
    \eta_i ~&=~ \sum_{(p,r,a_{pr})\in {\cal A}_i} a_{pr},\\
    \forall p\in P_i:~~~~f(p,v_i) ~&=~ \sum_{r\in R_i:~ (p,r,a_{pr})\in {\cal A}_i} a_{pr},\\
    \forall r\in R_i:~~~~f(v_i,r) ~&=~ \sum_{p\in P_i:~ (p,r,a_{pr})\in {\cal A}_i} a_{pr}.
  \end{align*}
  At each iteration we get at least one new element $u\in P_i\cup R_i$ with $\ell(u)=0$.
  Therefore we have $|{\cal A}_i|\le |P_i|+|R_i|-1$ because the last subset of $u \in P_i \cup R_i$ such that $\ell(u)$ turns zero must have cardinality at least two.  Consider ${\cal A} = \bigcup_{i\in I}{\cal A}_i$ which has a value of
  \[
  \sum_{(p,r,a_{pr})\in {\cal A}} a_{pr} ~=~
  \sum_{i\in I} \sum_{(p,r,a_{pr})\in {\cal A}_i} a_{pr} ~=~ 
  \sum_{i\in I} \eta_i ~=~	\eta.
  \]
  For each $p\in P$, we use that the directed edge $(s,p)$ has capacity $s_p$
  to conclude that 
  \[
  \sum_{(p,r,a_{pr})\in {\cal A}} a_{pr} ~=~
  \sum_{i\in I} \left( \sum_{(p,r,a_{pr})\in {\cal A}_i} a_{pr}\right) ~=~
  \sum_{i\in I} f(p,v_i) ~=~
  f(s,p) ~\le~ s_p.
  \]
  A similar argument using that the capacity of the directed edge $(r,t)$ is $d_r$, 
  for all $r\in R$, shows that
  \[
  \forall r\in R:~~~~ \sum_{p\in P:~ (p,r,a_{pr})\in {\cal A}} a_{pr} ~=~ f(r,t)~\le~ d_r.
  \]
  
  To bound the size of ${\cal A}$ we note that
  \[
  |{\cal A}| ~=~ \sum_{i\in I}|{\cal A}_i| ~\le~ \sum_{i\in I}\bigl(|P_i|+|R_i|-1\bigr)
  ~\le~ \sigma.	
  \]
  
  If the compact representation is edge-disjoint, then ${\cal A}$ is
  already a matching. Otherwise, it can happen that there are multiple
  triples of the form $(p,r,\cdot)$ for the same $(p,r)\in P\times
  R$. We handle this by sorting the elements of ${\cal A}$ with
  respect to the first two entries and merging all elements for a
  pair~$(p, r)$.  This takes $O(n+|{\cal A}|)= O(\sigma)$ time using
  radix sort.
\end{proof}

\begin{theorem}
  \label{thm:integral}
  Consider a geometric matching problem for a set $P$ of at most $n$
  points and a set $R$ of at most $n$ ranges such that the demands and
  the supplies are integral.  Let $\mu$ be the target value of the
  instance.  Assume that we have a compact representation of the
  incidence graph $I(P,R)$ of size $\sigma$.  Then we can compute a
  maximum (many-to-many) matching of size $O(\sigma)$ in
  $O(\sigma^{1+\eps} \log \mu)$ time with high probability, for any
  constant~$\eps>0$.
\end{theorem}
\begin{proof}
  We build the max-flow instance associated with the compact
  representation in $O(\sigma)$ time.  Note that the resulting
  instance has $O(\sigma)$ vertices, $O(\sigma)$ edges, and maximum
  capacity~$\mu$.  We solve the $s$-$t$ max-flow problem in
  $O(\sigma^{1+o(1)} \log \mu)$ time with high probability using the
  algorithm of Chen et al.~\cite[Corollary 1.3]{ChenKLPGS22}.  This
  running time can be upper bounded\footnote{We phrase the time bound
  in this way because this expression is commonly used in
  computational geometry.} by~$O(\sigma^{1+\eps} \log \mu)$ for any
  constant~$\eps>0$.
\end{proof}


\subsection{Maximum and perfect matchings in intersection graphs}

We next turn our attention to the setting where we want to compute a
maximum one-to-one matching in a bipartite graph that has a compact
representation. A particular case of this is deciding whether the
graph has a perfect matching (and constructing one).

\newcommand{\maximumMatchingTheorem}{\begin{restatable}{theorem}{maximumMatching}
    \label{thm:maximum_matching}
    Let $P$ be a set of at most $n$ points and let $R$ be a set of at
    most $n$ ranges.  Assume that we have a compact representation of
    the incidence graph $I(P,R)$ of size $\sigma$.  We can compute a
    maximum one-to-one matching in~$I(P,R)$ in $O(\sigma^{1+\eps})$
    time with high probability, for any constant $\eps>0$.
  \end{restatable}}

This is essentially a matching problem with unit supplies and demands.
Standard integrality results imply that the value of the matching is
the size of the maximum one-to-one matching.  However, we have to be
careful because the algorithm in Theorem~\ref{thm:integral} could
return a \emph{non-integral} matching.\footnote{It is not
clear to us whether the algorithms of Chen et al.~\cite{ChenKLPGS22}
guarantee that the flow is integral when the capacities are integral.
In any case, this potential problem appears if using any max-flow
algorithm that does not guarantee integrality.}
We next explain how to convert a maximum non-integral matching into a
maximum one-to-one matching; this is mostly an adaptation of the
method by M{\k{a}}dry~\cite[Section 8 of the full version]{Madry13},
which uses the algorithm of Goel, Kapralov and Khanna~\cite{GoelKK13}
to compute perfect matchings in regular bipartite graphs.

We first discuss the simpler case of perfect matchings.

\begin{lemma}
  \label{lem:perfect_matching}
  Let $P$ be a set of $n$ points and let $R$ be a set of $n$ ranges.
  Assume that we have a compact representation of the incidence graph
  $I(P,R)$ of size $\sigma$.  We can decide whether $I(P,R)$ has a
  perfect matching (and return one if it exists) in $O(\sigma^{1+\eps})$ time with
  high probability, for any constant $\eps>0$.
\end{lemma}
\begin{proof}
  We consider the matching problem where each point $p\in P$ has
  unit supply and each range $r\in R$ has unit demand. The matching
  problem has target value $\mu=n$.  We use Theorem~\ref{thm:integral} to
  decide whether the matching instance can be satisfied in
  $O(\sigma^{1+\eps/2} \log \mu)= O(\sigma^{1+\eps/2} \log n)=
  O(\sigma^{1+\eps})$ time with high probability, for any constant
  $\eps>0$.

  If there is no satisfying matching, then $I(P,R)$ has no perfect
  matching.  If it has a satisfying matching, we get from
  Theorem~\ref{thm:integral} one of size $O(\sigma)$. This satisfying
  matching is a so-called fractional perfect matching, which can
  also be interpreted as a doubly-stochastic matrix on $P\times R$.

  Goel, Kapralov and Khanna~\cite[Theorem 5]{GoelKK13} show how to
  convert a fractional perfect matching into a perfect matching in
  $O(m+n\log^2 n)$ expected time, where $m$ is the number of edges
  appearing in the fractional matching and $n$ is the number of
  vertices. In our setting we have $O(\sigma)$ edges, which means that
  we spend $O(\sigma + n\log^2 n)= O(\sigma^{1+\eps})$ expected time
  to obtain the perfect matching from the satisfying matching.  To
  guarantee high probability, we can run $O(\log n)$ copies of the
  procedure in parallel until the first on finishes, or use the
  version by Goel, Kapralov and Khanna~\cite[Section 2.3]{GoelKK13}
  that guarantees high probability.
\end{proof}

\maximumMatchingTheorem
\begin{proof}
  We consider the matching problem where each point $p\in P$ has
  unit supply and each range $r\in R$ has unit demand. For this
  matching problem we have $\mu=\min\{ |P|,|R|\}$.
	
  We use Theorem~\ref{thm:integral} to find a maximum matching in
  $O(\sigma^{1+\eps/2} \log \mu)= O(\sigma^{1+\eps/2} \log n)=
  O(\sigma^{1+\eps})$ time with high probability, for any constant
  $\eps>0$.  Let $\eta$ be the value of the matching. Because the
  demands and the supplies are integral, $\eta$ is integral and it is
  the size of the maximum matching in $I(P,R)$.
  \begin{figure}
    \centerline{\includegraphics[scale=1.1,page=4]{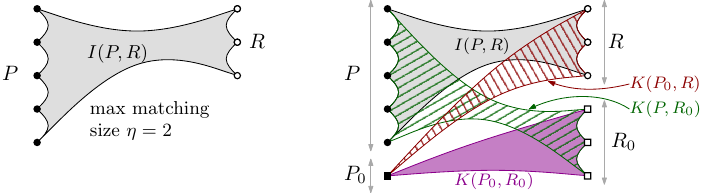}}
    \caption{Reducing a bipartite maximum matching instance in $I(P,R)$ (left)
      to a bipartite perfect matching instance in $I(P',R')$ (right).}
    \label{fig:maximum}
  \end{figure}
  
  We now create an abstract incidence graph, as follows; see
  Figure~\ref{fig:maximum}. We create a set $P_0$ of $|R|-\eta$ new
  vertices and a set $R_0$ of $|P|-\eta$ new vertices.  The bipartite
  classes are $P'=P\cup P_0$ and $R'=R\cup R_0$. Both have exactly
  $|P|+|R|-\eta$ elements.  Besides the incidences of $I(P,R)$, we
  declare that there is an incidence between each element of~$P_0$ and
  each element of~$R'$, and that there is an incidence between each
  element of~$P'$ and each element of~$R_0$.  Equivalently, we add the
  bipartite graphs $K(P,R_0)$, $K(P_0,R)$ and $K(P_0,R_0)$ to~$I(P,
  R)$.  This finishes the description of the new, abstract incidence
  graph~$I(P',R')$.  Note that $I(P',R')$ has a compact representation
  of size $\sigma+(|P|+|R_0|)+(|P_0|+|R|)+(|P_0|+|R_0|) = \sigma +
  O(n) = O(\sigma)$, because we are just adding three complete
  bipartite graphs on at most $2n$ vertices.
	
  The graph $I(P',R')$ has a perfect matching, that is, a matching of
  size $|P|+|R|-\eta$.  Indeed, we can extend the maximum matching of
  $I(P,R)$ with value $\eta$ to a matching of value $|P|+|R|-\eta$ by
  sending any remaining supply from $p\in P$ to some of the new ranges
  of $R_0$.  Similarly, each unsatisfied demand from $r\in R$ can be
  provided by some of the new points in $P_0$. Since we have the
  complete bipartite graphs $K(P,R_0)$, $K(P_0,R)$, we can do this
  arbitrarily. Finally, any remaining supply from $P_0$ and demand
  from $R_0$ can be assigned arbitrarily within $K(P_0,R_0)$.  This
  shows that there is a matching for~$I(P',R')$ with
  value~$|P|+|R|-\eta=|P'|=|R'|$, and therefore $I(P',R')$ has a
  perfect matching.

  We use Lemma~\ref{lem:perfect_matching} for~$I(P',R')$ and we get a
  perfect matching with~$|P|+|R|-\eta$ edges.  We remove the edges
  that are incident to~$P_0$ and~$R_0$, and are left with
  $(|P|+|R|-\eta)- |P_0|-|R_0| = (|P|+|R|-\eta)- (|R|-\eta) -
  (|P|-\eta) =\eta$ edges that connect~$P$ and~$R$. This is a maximum
  one-to-one matching in~$I(P,R)$.
\end{proof}

We can now combine Theorem~\ref{thm:maximum_matching} and
Theorem~\ref{thm:point+rectangle} to obtain the following.
\begin{corollary}
  \label{coro:maximum+points+rectangles}
  Let $P$ be a set of at most $n$ points and let $R$ be a set of at
  most $n$ axis-parallel boxes in $\Reals^d$, for constant $d$.  We
  can compute a maximum one-to-one matching in the incidence
  graph~$I(P,R)$ in time $O(n^{1+\eps})$ with high probability, for
  any constant~$\eps>0$.
\end{corollary}

Theorem~\ref{thm:maximum_matching} and Theorem~\ref{thm:point+disk1} lead to the following.
\begin{corollary}
  \label{coro:maximum+points+disks}
  Let $P$ be a set of at most $n$ points and let $D$ be a set of at
  most $n$ congruent disks in the plane.  We can compute a maximum
  one-to-one matching in the incidence graph~$I(P,D)$ in
  time~$O(n^{4/3+\eps})$ with high probability, for any
  constant~$\eps>0$.
\end{corollary}

Theorem~\ref{thm:maximum_matching} and Theorem~\ref{thm:point+halfspace} lead to the following.
\begin{corollary}
  \label{coro:maximum+points+halfspaces}
  Let $P$ be a set of at most $n$ points and let $H$ be a set of at
  most $n$ halfspaces in $\Reals^d$, for constant $d$.  We can compute
  a maximum one-to-one matching in the incidence graph~$I(P,H)$ in
  time $O(n^{\frac{2d}{d+1}+\eps})$ with high probability, for any
  constant~$\eps>0$.
\end{corollary}


\subsection{Bottleneck distances}
\label{sec:bottleneck}

Let $P$ and $Q$ be two sets of $n$ points in the plane and consider
the bipartite graph obtained by connecting points whose
$L_\infty$-distance is at most $\lambda$:
\begin{align*}
  G_\infty(P,Q,\lambda) ~&=~ \big(P\sqcup Q, \{pq \mid \max\{|x(p)-x(q)|,|y(p)-y(q)|\} \le \lambda \}\big).
\end{align*}
We are interested in computing the \emph{$L_\infty$-bottleneck
distance} between $P$ and $Q$, defined as
\begin{align*}
  \lambda^*_\infty(P,Q) ~&=~ \min \{ \lambda \mid
  G_\infty(P,Q,\lambda) \text{ has a perfect matching} \}.
\end{align*}
The same concepts can be defined for the $L_1$- and $L_2$-distances
\begin{align*}
  G_1(P,Q,\lambda) ~&=~ \big(P\sqcup Q, \{pq \mid |x(p)-x(q)|+ |y(p)-y(q)| \le \lambda \}\big),\\
  \lambda^*_1(P,Q) ~&=~ \min \{ \lambda \mid 
  G_1(P,Q,\lambda) \text{ has a perfect matching}\},\\
  G_2(P,Q,\lambda) ~&=~ \big(P\sqcup Q, \{pq \mid (x(p)-x(q))^2+(y(p)-y(q))^2 \le \lambda^2 \}\big),\\
  \lambda^*_2(P,Q) ~&=~ \min \{ \lambda \mid 
  G_2(P,Q,\lambda) \text{ has a perfect matching}\}.
\end{align*}
The purpose of this section is to compute $\lambda^*_\infty(P,Q)$, $\lambda^*_1(P,Q)$
and $\lambda^*_2(P,Q)$ effectively.

Since disks in the $L_\infty$- and $L_1$-metrics are squares, either
axis-parallel or rotated, we can use
Corollary~\ref{coro:maximum+points+rectangles} to obtain the
following.

\begin{corollary}
  \label{coro:L_1+decision}
  Let $P$ and $Q$ be two sets of $n$ points in the plane and let
  $\lambda>0$ be a given value. We can decide whether the graphs
  $G_\infty(P,Q,\lambda)$ and $G_1(P,Q,\lambda)$ have a perfect
  matching in $O(n^{1+\eps})$ time with high probability, for
  any~$\eps>0$.
\end{corollary}
\begin{proof}
  Let $D_\infty(p,\lambda)$ be the disk of radius $\lambda$ centered
  at $p$ in the $L_\infty$-metric.  Note that $D_\infty(p,\lambda)$ is
  an axis-aligned square.  Then, the graph $G_\infty(P,Q,\lambda)$ has
  a perfect matching if and only if the incidence graph between the
  points $P$ and the axis-parallel square ranges $\{
  D_\infty(q,\lambda)\mid q\in Q \}$ has a perfect matching.
	
  For $G_1(P,Q,\lambda)$ and the $L_1$-metric, the same argument
  applies because the disks in the $L_1$-metric are squares that
  become axis-parallel after a rotation by $\pi/4$.
\end{proof}

As noted by Efrat, Itai and Katz~\cite[Section 6.2.2]{efrat} we can
now perform a binary search to find $\lambda^*_\infty(P,Q)$.  We
provide an overview of the idea.  Let $x_1,\dots,x_n$ be the sorted
$x$-coordinates of $P$ and let $x'_1,\dots,x'_n$ be the sorted
$x$-coordinates of $Q$.  Similarly, let $y_1,\dots,y_n$ be the sorted
$y$-coordinates of $P$ and let $y'_1,\dots,y'_n$ be the sorted
$y$-coordinates of $Q$.  We then construct the four $n\times n$
matrices $D_x,\bar D_x, D_y,\bar D_y$ with entries ($1\le i,j \le n$)
\begin{align*}
  D_x(i,j) = x_i-x'_j, ~~~~~& \bar D_x(i,j) = x'_i-x_j,\\
  D_y(i,j) = y_i-y'_j, ~~~~~& \bar D_y(i,j) = y'_i-y_j.
\end{align*}
These matrices are sorted, meaning that each row and each column is
monotone increasing or decreasing.  The optimal value
$\lambda^*_\infty(P,Q)$ is an element of the set
\[
\big\{\max\{|x(p)-x(q)|,|y(p)-y(q)|\} \mid p\in P, q\in Q\big\},
\] 
which is a
subset of the entries in the matrices ${\cal D}=\{D_x,\bar D_x,
D_y,\bar D_y\}$.  Thus, it suffices to find the smallest
entry~$\lambda$ in~${\cal D}$ such that $G_\infty(P,Q,\lambda)$ has a
perfect matching.  Selecting the $k$th element among the entries of
${\cal D}$ can be done using the algorithm of Frederickson and
Johnson~\cite{FredericksonJ84} in~$O(n)$ time.  It follows that the
binary search for the target value $\lambda^*_\infty(P,Q)$ requires
$O(\log (4n^2))=O(\log n)$ iterations, where in each iteration we
spend $O(n)$ time to select some $k$th value of ${\cal D}$ and a use
of Corollary~\ref{coro:L_1+decision}.  The logarithmic term is
absorbed by the $O(n^{\eps})$-term.

A similar argument applies to the $L_1$-metric.  We summarize.
\begin{corollary}
  \label{coro:L_1+bottleneck}
  Let $P$ and $Q$ be two sets of $n$ points in the plane.  We can compute
  the bottleneck distances $\lambda^*_\infty(P,Q)$ and
  $\lambda^*_1(P,Q)$ in $O(n^{1+\eps})$ time with high probability,
  for any~$\eps>0$.
\end{corollary}

The time bounds of Corollary~\ref{coro:L_1+bottleneck} also apply 
to the $L_\infty$~metric in any fixed dimension~$d$. 
The decision problem can be solved using
the incidence graph of points and axis-parallel cubes
because of Corollary~\ref{coro:maximum+points+rectangles}.
The binary search to find the optimal value is an extension 
of the approach we have described for~$\Reals^2$, 
but using $O(d)$ sorted matrices. See
Efrat, Itai and Katz~\cite[Section 6.2.2]{efrat} for an explicit description.

A similar approach can be used for the $L_2$-distance.  Deciding
whether $G_2(P,Q,\lambda)$ has a perfect matching reduces to deciding
whether the incidence graph of a set of points and a set of disks of
radius $\lambda$ has a perfect
matching. Corollary~\ref{coro:maximum+points+disks} readily implies
the following.
\begin{corollary}
  \label{coro:L_2+decision}
  Let $P$ and $Q$ be two sets of $n$ points in the plane and let
  $\lambda>0$ be a given value. We can decide whether the graph
  $G_2(P,Q,\lambda)$ has a perfect matching in $O(n^{4/3+\eps})$ time
  with high probability, for any~$\eps>0$.
\end{corollary}

We can now perform a binary search in the set~$D$ of pairwise
distances between the points~$P$ and the points~$Q$, which
contains~$\lambda^*_2(P,Q)$.  Selecting the $k$th distance in $D$
takes $O(n^{4/3}\log^2 n)$ time using the algorithm of Katz and
Sharir~\cite{katz-1997}, or $O(n^{4/3}\log n)$ time using 
the recent improvement of Wang and Zhao~\cite{WangZ23}.  
Thus, each of the $O(\log (n^2))=O(\log n)$
iterations of the binary search takes $O(n^{4/3}\log n)$ time to
select some $k$th distance and a use of
Corollary~\ref{coro:L_2+decision}.  Again, the logarithmic term is
absorbed by the $O(n^{\eps})$-term.  We summarize.
\begin{corollary}
  \label{coro:L_2+bottleneck}
  Let $P$ and $Q$ be sets of $n$ points in the plane.  We can compute
  the bottleneck distance $\lambda^*_2(P,Q)$ in $O(n^{4/3+\eps})$ time
  with high probability, for any~$\eps>0$.
\end{corollary}


\subsection{Bottleneck distance of persistence diagrams}
\label{sec:persistence}

We next describe an application to compute the bottleneck distance of
\emph{persistence diagrams}, one of the most important concepts in
computational topology; see for example the book by Edelsbrunner and
Harer~\cite[Part C]{EH} or the book by Dey and Wang~\cite[Chapter
  3]{DW}.

Let $\mathbb{X}$ be a topological space and consider a function
$f\colon\mathbb{X}\rightarrow \Reals$.  For any dimension $k$, the
persistence diagram $\Dgm_k(f)$ at dimension $k$ represents the
evolution of the $k$-th dimensional homology groups through the
filtrations $f^{-1}((-\infty,t])$ as $t\in \Reals$ increases.  It is a
  multiset of points in the plane~$\Reals^{2}$ above the diagonal
  $\Delta=\{(x,x)\mid x\in \Reals\}$ together with all the points
  on the diagonal~$\Delta$.  The persistence diagram is uniquely
  defined by the points above the diagonal.

The bottleneck distance between two persistence diagrams~$X$ and~$Y$
is defined as
\[
	W_\infty(X,Y) = \adjustlimits \inf_{\varphi:X\rightarrow Y} \sup_{p\in X}\| p-\varphi(p) \|_\infty,
\]
where $\varphi$ iterates over all the bijections. Here it is important
that~$X$ and~$Y$ have an infinite number of points because of the
added diagonal~$\Delta$, and thus there are several bijections.  The
celebrated stability of persistence diagram 
tells us that, for a triangulable topological space~$\mathbb{X}$, and tame functions
$f,g\colon\mathbb{X}\rightarrow \Reals$, the bottleneck distance of
the persistence diagrams of $f$ and $g$ is upper bounded by the
$L_\infty$-norm of $f-g$, that is,
\[
\forall k: ~~~ W_\infty(\Dgm_k(f),\Dgm_k(g)) ~\le~ \|f - g \|_\infty ~=~ \max_{x\in \mathbb{X}} |f(x)-g(x)|.
\]
See the original work of Cohen-Steiner, Edelsbrunner, and
Harer~\cite{Cohen-SteinerEH07} or the textbook~\cite[Part C]{EH}.  It
is hard to overstate the importance of stability in modern topological
data analysis, and the computation of the bottleneck distance between
persistence diagrams is one of the key tools. See for example Kerber,
Morozov and Nigmetov~\cite{KerberMN17} for a library offering an
efficient computation.

We now turn to our result. First we discuss a characterization of the
bottleneck distance using discrete point sets. The main idea is that
for a point outside the diagonal~$\Delta$ we can decide to match it to
a point outside the diagonal or to its orthogonal projection
on~$\Delta$.

Let~$X$ and~$Y$ be two persistence diagrams.  Let $X_0$ be the points
of $X$ outside the diagonal $\Delta$ and let $X'_0$ be the point set
obtained by projecting $X_0$ onto $\Delta$.  Thus, $X'_0=\{
(\tfrac{x(p)+y(p)}{2},\tfrac{x(p)+y(p)}{2})\mid p\in X_0\}$.  We
define $Y_0$ and $Y'_0$ analogously.  Set $P=X_0\sqcup Y'_0$ and
$Q=Y_0\sqcup X'_0$. Note that $|P|=|Q|= |X|+|Y|$.  For
each~$\lambda\ge0$, we define the bipartite graph $G_\lambda(X,Y)$
with vertex set $P\sqcup Q$ and edge set
\[
\{ pq \mid p\in Y'_0 \text{ and } q\in X'_0 \} ~\cup~
\{ pq \mid (p\in X_0 \text{ or } q\in Y_0) \text{ and } \| p-q \|_\infty \le \lambda\}. 
\]
Thus, we add the edges connecting points at $L_\infty$-distance at
most $\lambda$ and also add all the edges connecting points on the
diagonal $\Delta$.  The following property is a standard observation;
see for example~\cite[Section VIII.4]{EH} or~\cite[Lemma
  2.2]{KerberMN17}.

\begin{lemma}
  The bottleneck distance of persistence diagrams~$X$ and~$Y$ is the
  minimum~$\lambda\ge 0$ such that~$G_\lambda(X,Y)$ has a perfect
  matching.
\end{lemma}

The decision version is then easy.
\begin{lemma}
  Let $X$ and $Y$ be persistence diagrams, each with at most $n$
  points outside the diagonal and let~$\lambda$ be a given value.  We
  can decide whether~$G_\lambda(X,Y)$ has a perfect matching
  in~$O(n^{1+\eps})$ time with high probability, for any~$\eps>0$
\end{lemma}
\begin{proof}
  The graph $G_\lambda(X,Y)$ has a compact representation because it
  is $G_\infty(P,Q,\lambda)\cup K(Y'_0,X'_0)$.  The first part is
  equivalent to an incidence graph between points and squares of side
  length~$\lambda$, which has a compact representation of size
  $O(n\log^2 n)$ by Theorem~\ref{thm:point+rectangle}.  
  The second part is a complete graph with a compact
  representation of size $|X'_0|+|Y'_0|=|X|+|Y|\le 2n$. Altogether,
  the graph $G_\lambda(X,Y)$ has a compact representation of size
  $O(n\log^2 n)$ and we can decide whether it has a perfect matching
  using Theorem~\ref{thm:maximum_matching}.
\end{proof}

The bottleneck distance of two persistence diagrams~$X$ and~$Y$ is
defined by the $L_\infty$-distance between a point in $P=X_0\sqcup
Y'_0$ and a point in $Q=Y_0\sqcup X'_0$. This is the case because only
for those values of $\lambda$ the graph $G_\lambda(X,Y)$ undergoes a
change.  The same method that has been described for the binary search
of Corollary~\ref{coro:L_1+bottleneck}, namely using searching in
sorted matrices, applies here. We obtain the following:

\begin{corollary}
  For persistence diagrams~$X$ and~$Y$, each with at most $n$
  points outside the diagonal, we can compute their bottleneck
  distance $W_\infty(X,Y)$ in $O(n^{1+\eps})$ time with high
  probability, for any~$\eps>0$
\end{corollary}


\section{Geometric matching problems with arbitrary supplies and demands}
\label{sec:arbitrary}

In the previous section we made use of a fast algorithm for max-flow
that works for integer capacities.  In this section we consider the
scenario where supplies and demands are arbitrary real numbers, and we
insist on an algorithm that solves the problem exactly on the real
RAM, as is customary in computational geometry.

Clearly, we can restate Theorem~\ref{thm:integral} by making use of a
combinatorial max-flow algorithm.  The fastest combinatorial algorithm
that works for arbitrary capacities runs in time~$O(n m)$, where~$n$
is the number of vertices, $m$~the number of edges of the
network. Applying this to the graph of Figure~\ref{fig:example2}
results in a running time of~$O(\sigma^{2})$, where~$\sigma$ is the
size of the compact representation of the incidence graph~$I(P, R)$.

We now show how this can be improved when~$\sigma > n$.  Instead of
explicitly computing the five-layer graph of Figure~\ref{fig:example2}, 
we run Dinitz' algorithm on the classic four-layer graph
based on~$I(P,R)$, but we represent that graph and all intermediate
graphs implicitly, making use of compact representations.

\subsection{Dinitz' algorithm with pruning}

Given a set of at most $n$~points~$P$ with supplies~$s_p$, a set of at
most $n$~ranges~$R$ with demands~$d_r$, and an incidence graph~$I(P,
R)$ given through a compact representation of size~$\sigma$, we
consider the flow network consisting of a source~$s$, the
vertices~$P$, the vertices~$R$, and a sink~$t$: For each~$p\in P$, we
have the edge~$(s, p)$ with capacity~$s_{p}$, for each~$r\in R$, we
have the edge~$(r, t)$ with capacity~$d_{r}$, and for each~$(p, r) \in
I(P, R)$ we have the edge~$(p, r)$ with infinite capacity. Let~$G$
denote this flow network.  Since~$G$ can have quadratic
complexity, we will not actually construct it explicitly.

Let~$f$ be a flow in~$G$. Naturally, $f$~induces a matching, so we
call~$f$ a flow or matching interchangeably.  Let $H(f)$ be the
\emph{undirected} graph obtained from~$G$ by removing~$s$ and~$t$, and
keeping the undirected edges $\{p,r\}$ if and only if~$f(p,r) > 0$.
\begin{lemma}
  \label{lem:forest}
  For every flow $f$ in $G$, there exists a flow $f_0$ in $G$ with the
  same value such that $H(f_0)$ is a forest and a subgraph of $H(f)$.
\end{lemma}
\begin{proof}	
  Suppose that $H(f)$ contains a cycle.  This cycle zig-zags back and
  forth between~$P$ and~$R$.  We color the edges of the cycle
  alternatingly red and blue.  Both red and blue edges carry positive flows 
  out of their endpoints in $P$.  We can add a value~$\Delta$ to the flow
  on all red edges and subtract~$\Delta$ from the flow on all blue
  edges without changing the outgoing and incoming flow at any node.
  Let~$e$ be the blue edge that has the minimum flow value among the
  blue edges.  By choosing~$\Delta$ to be the flow value on~$e$, we
  can reduce the flow value on~$e$ to zero.  This gives a new flow
  $f'$ in $G$ such that $H(f')$ is the subgraph of $H(f)$ with~$e$
  (and possibly more edges) deleted.  We repeat this procedure until
  there is no cycle left in~$H(f_{0})$.
\end{proof}

We will use Dinitz's algorithm to compute a maximum flow in~$G$.
Dinitz' algorithm starts with the zero flow, and augments the flow in
phases.  In each phase, it computes the \emph{level graph} $L(f)$ of
the residual graph~$G(f)$ of~$G$ with respect to the current flow~$f$.
The level graph~$L(f)$ is defined as follows: For each node~$v$
of~$G(f)$, we define the \emph{level}~$\ell(v)$ of~$v$ as the length
of the shortest path (in terms of the number of edges) from~$s$ to~$v$ in~$G(f)$.  
The level graph contains exactly those edges~$(u, v)$
of~$G(f)$ where~$\ell(v) = \ell(u) + 1$.  In other words, all edges of
the level graph are on a shortest path starting in~$s$.

Clearly, $s$ is the only node at level zero in~$L(f)$. In our
application, every odd level consists of elements of~$P$, every even
level other than zero consists of elements of~$R$.  The sink~$t$
resides in an odd level.  We ignore all other nodes on
level~$\ell(t)$ and above.

Next, Dinitz' algorithm computes a \emph{blocking flow}~$g$ in
$L(f)$, and augments~$f$ to~$f + g$ before proceeding to the next
phase. A flow $g$ is blocking iff for every
path~$\gamma$ from $s$ to~$t$ in $L(f)$, there is an edge $e$
of~$\gamma$ that is saturated by~$g$, that is, $g$ is equal to the
capacity on $e$. Dinitz' algorithm terminates when there is no longer
a path from~$s$ to~$t$ in the level graph.

Our algorithm differs from Dinitz' algorithm in only one detail: After
augmenting the current flow, we apply Lemma~\ref{lem:forest} to prune
the current matching.  It follows that at the end of each phase, the
graph~$H(f)$ has at most~$2n-1$~edges (since it is a forest).

The following lemma shows that at most~$n$ phases of this algorithm
suffice to obtain a maximum flow in~$G$.  This is essentially
Theorem~8.4 in~\cite{tarjan}, but some extra effort is needed due to
the pruning step (by Lemma~\ref{lem:forest}) at the end of each phase.

We first establish some notation.  In the residual graph~$G(f)$ and
the level graph~$L(f)$, we call the edges from~$s$ to~$P$ the
\emph{feeder edges}, the edges from elements of~$P$ to elements of~$R$
the \emph{forward edges}, the edges from~$R$ to $t$ the \emph{draining
edges}, and all other edges \emph{backward edges}.  Note that backward
edges do not exist in~$G$, rather they are the reverse of an edge with
positive flow in~$G$.  In the residual graph, backward edges can be
directed from~$R$ to~$P$, from~$P$ to~$s$, or from~$t$ to~$R$.  Only
backward edges from~$R$ to~$P$ will appear in the level graph, as
$\ell(s) = 0$ and since we ignore levels after~$\ell(t)$.

By the definition of the residual graph, the capacity of every feeder
edge~$(s,p)$ is the unused supply left at~$p$, that is $s_p - \sum_{r
  \in R} f(p,r)$.  The capacities of forward edges are $\infty$.  The
capacity of every drain edge~$(r,t)$ is equal to the demand at~$r$
that has not been satisfied yet, that is~$d_r - \sum_{p \in P}
f(p,r)$.  Finally, the capacity of a backward edge $(u, v)$ is equal
to~$f(v, u)$.

\begin{lemma}
  \label{lem:block}
  Starting with the zero flow, at most $n$~successive augmentations by
  blocking flows in the corresponding level graphs will yield a
  maximum flow in~$G$.
\end{lemma}
\begin{proof}
  We will show that the level of $t$ in the level graph increases
  after each augmentation.  Since the level of~$t$ is always odd,
  starts with~3, and can be at most~$2n+1$, there can be at most~$n$
  augmentations.

  Let $f$ be the current flow in $G$.  Let $Z = G(f)$ denote the
  residual graph of~$G$ with respect to~$f$.  Let~$g$ be a blocking
  flow in $L(f)$.  We use $\ell(v)$ to denote the level
  of the node $v$ in~$L(f)$, which is also the distance of~$v$ from~$s$ in~$Z$.
  Define~$h = f + g$.  Let $f'$ be obtained from $h$
  by Lemma~\ref{lem:forest}.  Let $Z' = G(f')$ be the residual graph
  of~$G$ with respect to~$f'$.  
  \begin{claim}
    \label{claim:rr}
    If $(v, w)$ is a feeder, forward, or draining edge of~$Z'$,
    then~$(v, w)$ is an edge of~$Z$.
  \end{claim}
  \begin{claimproof}[Proof of Claim~\ref{claim:rr}]
    If $(v,w)$ is a forward edge in~$Z'$, then $(v,w)$ is an edge
    of~$G$. Since the capacity of~$(v,w)$ is infinite, it is also an
    edge of~$Z$.
    
    The feeder and draining edges are never involved in a cycle in the
    proof of Lemma~\ref{lem:forest}.  So the flows on them are the
    same with respect to~$f$ and~$f'$, which means that their residual
    capacities in~$Z$ and~$Z'$ are the same.  Therefore, if $(v, w)$
    is in~$Z'$, it is also in~$Z$.
  \end{claimproof}
  \begin{claim}
    \label{claim:1}
    For every edge $(v, w)$ in $Z'$, we have $\ell(w) \leq \ell(v) + 1$.
  \end{claim}
  \begin{claimproof}[Proof of Claim~\ref{claim:1}]
    We first observe that if~$(v, w)$ is an edge of~$Z$, then $\ell(w)
    \leq \ell(v) + 1$ follows from the definition of~$\ell(w)$. For
    feeder, forward, and draining edges the claim therefore follows
    immediately from Claim~\ref{claim:rr}.

    If $(v, w)$ is a backward edge, then $f'(w, v) > 0$.  By
    Lemma~\ref{lem:forest}, $H(f')$ is a subgraph of $H(h)$
    and so $h(w, v) > 0$.  If $f(w, v) > 0$ as well, then the edge
    $(v, w)$ appeared in~$Z$, so $\ell(w) \leq \ell(v) + 1$.

    If $f(w, v) = 0$, then $g(w, v) = h(w, v) > 0$.  The positive
    value of the blocking flow~$g(w,v)$ means that $(w,v)$ was an edge
    in~$L(f)$, which by definition of~$L(f)$ means that $\ell(v) =
    \ell(w) + 1$ and therefore~$\ell(w) = \ell(v) - 1 < \ell(v) + 1$.
  \end{claimproof}
  
  Let now $\ell'(v)$ be the level of a node $v$ in $Z'$, that is, the
  level of~$v$ in~$L(f')$.  We claim that $\ell'(v) \geq
  \ell(v)$.  Indeed, consider any path from $s$ to $v$ in~$Z'$.  By
  Claim~\ref{claim:1}, $\ell(\cdot)$ increases by at most one along
  each edge of this path, so~$\ell(v)$ is at most the length of this
  path.

  In particular, we have $\ell'(t) \geq \ell(t)$.  We show that in
  fact $\ell'(t) > \ell(t)$.  For the sake of contradiction, assume
  that $\ell'(t) = \ell(t)$.  This means that for any shortest path
  from $s$ to $t$ in~$Z'$, $\ell(\cdot)$ increases by exactly one on
  each edge of this path.  This means that for each edge $(v, w)$ on
  this path, $v$ and~$w$ are on consecutive levels of~$L(f)$.  We claim
  that any shortest path of length $\ell(t)$ from $s$ to $t$ in $L(f')$ 
  is also a path in $L(f)$.  The reason is as follows.
  By Claim~\ref{claim:rr}, if $(v,w)$ is a feeder, forward, or draining
  edge, then $(v,w)$ is an edge of~$Z$ and therefore also an edge
  of~$L(f)$.  In the remaining case, $v \in R$, $w \in P$,
  and~$\ell(w) = \ell(v) + 1$.  If $(v,w)$ is an edge in $L(f)$, we 
  are done.   If $(v,w)$ is not an edge of~$L(f)$,
  then~$f(w,v) = 0$.  Since $(v,w)$ is an edge of~$Z'$, we have
  $f'(w,v) > 0$, and by Lemma~\ref{lem:forest} $h(w,v) = g(w,v) > 0$.
  This implies again that $(w,v)$ must have been a forward edge
  in~$L(f)$, and so $\ell(v) = \ell(w) + 1$, a contradiction.  This establishes
  our claim.

  It follows that the entire shortest path of length $\ell(t)$
  from~$s$ to~$t$ was already a path in~$L(f)$.  But by definition of
  a blocking flow, that means that~$g$ must have saturated some edge
  along this path.  If this was a feeder or draining edge, those edges
  are no longer in~$Z'$, a contradiction.  A forward edge cannot be
  saturated (it has infinite capacity), so it must be a backward edge,
  say $(r,p)$.  Saturating a backward edge~$(r,p)$ means that $g(r,p)
  = f(p,r)$.  It follows that $h(p,r) = 0$, which implies $f'(p,r) =
  0$ by Lemma~\ref{lem:forest}.  Therefore, $(r,p)$~is not an edge
  of~$Z'$, again a contradiction.
\end{proof}

We cannot afford to store the residual graph~$G(f)$ or the level
graph~$L(f)$ explicitly, as they may have a quadratic number of edges.
We will therefore store these graphs implicitly, using the given
compact representation for~$I(P, R)$.  Inbetween phases of the
algorithm, we only need to store the current flow~$f$.  Since it is
non-zero on~$O(n)$ edges only, we store this flow explicitly.  In the
following sections we describe how to implement a phase of Dinitz'
algorithm efficiently.

\subsection{Constructing the level graph}

Given the current flow~$f$, we will represent the level graph~$L(f)$
by explicitly constructing the elements of each level of~$L(f)$, all
feeder edges, all draining edges, and all backward edges.  There
are~$O(n)$ such edges, so we can indeed afford to store them
explicitly.  The forward edges are represented using compact
representations.  

Recall that we are given a compact representation~$\{ (U_i,V_i)\mid
i\in I\}$ for~$I(P, R)$.  We start by building an index for this
representation, which, for every~$p \in P$, gives us a list of all~$i
\in I$ where~$p \in U_i$.  We also set~$\tilde P := P$,~$\tilde R
:= R$, and~$\tilde I := I$.

There are at most $|P| \leq n$ feeder edges from $s$ to~$P$ in~$L(f)$.
We construct these edges and their capacities by brute-force in $O(n)$
time.  This also yields the set~$P_1 \subseteq P$ of elements of
level~1 of~$L(f)$.

We now query our index to retrieve the set of indices~$I_1 \subseteq
\tilde I$ such that~$U_i \cap P_1 \neq \emptyset$ for~$i \in I_1$.  We
set~$R_2 = \bigcup_{i \in I_1} V_i$ to be the set of vertices on
level~2, and build the following compact representation for the
forward edges from level~1 to level~2:
\[
\{ (U_i \cap P_1, V_i) \mid i\in I_1\},
\]
We update~$\tilde P := \tilde P \setminus P_1$,~$\tilde R := \tilde
R \setminus R_2$, and~$\tilde I := I \setminus I_1$.

If there is any~$r \in R_2$ that has an unmet demand, then we
place~$t$ in level~3 and produce the draining edges~$(r, t)$ with
capacity equal to the unmet demand of each~$r \in R_2$. The level
graph is then complete.

If there is no~$r \in R_2$ with unmet demand, we proceed to construct
level~3. For each~$r \in R_2$, we determine all~$p \in \tilde P$ with
flow into~$r$, that is, with~$f(p, r) > 0$.  The set~$P_3$ of all
these vertices~$p$ constitutes level~3. We record the edges~$(r,p)$ as
backward edges from level~2 to level~3.

We now again use our index to find the set of indices~$I_3 \subseteq
\tilde I$ such that~$U_i \cap P_3 \neq \emptyset$ for~$i \in I_3$.  We
set~$R_4 = \bigcup_{i \in I_3} V_i \cap \tilde R$ to be the set of
vertices on level~4, and build the following compact representation of
the forward edges from level~3 to level~4:
\[
\{ (U_i \cap P_3, V_i \cap \tilde R) \mid i\in I_3\},
\]
Again we update~$\tilde P := \tilde P \setminus P_3$,~$\tilde R :=
\tilde R \setminus R_4$, and~$\tilde I := \tilde I \setminus I_3$, and
repeat this process until either we place the sink~$t$ or an even
level becomes empty.  In the latter case there is no path from~$s$
to~$t$ in the level graph, and the current flow is a maximum flow.

\begin{lemma}
  \label{lem:level}
  The elements of each level of~$L(f)$, the feeder edges, the backward
  edges, the draining edges, and compact representations of all
  forward edges of $L(f)$ can be constructed in~$O(n + \sigma)$ time.
\end{lemma}
\begin{proof}
  Constructing the feeder and draining edges explicitly takes $O(n)$
  time.  Finding the backward edges takes~$O(n)$ total time as~$H(f)$
  has~$O(n)$~edges.  Each pair~$(U_i, V_i)$ of the compact
  representation of~$I(P, R)$, is handled only once, and this takes
  time~$O(|U_i| + |V_i|)$ using our index and suitable datastructures
  for~$\tilde P$,~$\tilde R$, and~$\tilde I$.
\end{proof}

\subsection{Constructing a blocking flow}

We use Sleator and Tarjan's algorithm~\cite{sleator} for finding a
blocking flow in the level graph, see for instance the presentation in
Section~8.2 of Tarjan's book~\cite{tarjan}.  The algorithm computes a
blocking flow on a level graph with $n$~vertices and $m$~edges in
time~$O((n+m)\log (n+m))$.  To run the algorithm, we construct a
concrete level graph~$L'$ from our implicitly represented level
graph~$L(f)$. We start by setting~$L'$ to contain all the vertices,
feeder, backwards, and draining edges of~$L(f)$.  For each pair of
layers~$P_j$ and~$R_{j+1}$, for $j = 1, 3, 5, \dots$, we then proceed
as in Section~\ref{sec:integral}.  Namely, we insert an intermediate
layer~$C_j$ as follows: Recall that we have a compact representation
$\{(U_i, V_i) \mid i \in I_j\}$ of the forward edges from~$P_j$
to~$R_{j+1}$.  For each $i\in I_j$, we add a new vertex $v^j_i$,
directed edges $(p,v^j_i)$ for all $p\in U_i$, and directed edges
$(v^j_i,r)$ for all $r\in V_i$.  These edges have infinite capacity.

The graph~$L'$ has $O(\sigma)$ vertices and edges.  It is easy to see
that~$L'$ is the level graph of a flow network (namely the flow
network~$L'$ itself), so Sleator and Tarjan's algorithm computes a
blocking flow~$g'$ for~$L'$ in time~$O(\sigma \log \sigma)$.

We convert the flow~$g'$ on~$L'$ into a flow~$g$ on the graph~$L(f)$
as we have done in Lemma~\ref{le:reduction}.  It remains to argue
that~$g$ is indeed a blocking flow on~$L(f)$.  To this end, we observe
that the flows~$g$ and~$g'$ are identical on all feeder, backwards,
and draining edges.  The forward edges, on the other hand, have
infinite capacity, so they cannot be saturated by any flow.  For any
node~$v \in L(f)$ and any path from~$s$ to~$v$ in~$L(f)$, there is
also a path from~$s$ to~$v$ in~$L'$ that uses the same finite-capacity
edges.  Since~$g'$ saturates an edge on this path, $g$~also saturates
an edge on the path from~$s$ to~$v$ in~$L(f)$.

\subsection{Pruning the matching}

After computing a blocking flow, we augment the previous flow.  Since
the previous flow is non-zero on~$O(n)$ edges and the blocking flow is
non-zero on~$O(\sigma)$ edges, this can be done in time~$O(n + \sigma)
= O(\sigma)$. It remains to describe how to prune this new flow~$f$
efficiently, that is, how to compute the flow~$f_{0}$ guaranteed by
Lemma~\ref{lem:forest}. This is important to start each iteration with
a flow that uses $O(n)$ edges.

This turns out to be surprisingly hard.
We will use a red-blue link-cut tree, a variant of link-cut trees with
the following properties: The data structure stores a collection of
rooted trees on nodes that are either red or blue.  Every edge
connects a red node and a blue node, and carries a non-negative cost
(there are no costs on the nodes).  We will consider edges to have the
same color as their endpoint closer to the root of the tree (the
``parent'' of the edge).  Since we will support changing the tree
roots, this means that the color of an edge is a \emph{dynamic}
property.

The data structure supports the following operations in time~$O(\log
n)$, where~$n$ is the number of nodes in the forest.
\begin{di}
\item $\opmaketree(v)$: Create a new tree containing the single
  node~$v$.  The color of $v$ is given, as part of $v$.
\item $\opfindroot(v)$: Return the root of the tree containing~$v$.
\item $\oplink(v,w,x)$: Combine the trees containing $v$ and $w$
  into a single tree by adding the edge~$vw$ of value~$x$,
  making~$w$ the parent of~$v$.  This operation assumes that $v$ is
  a tree root and that $v$ and~$w$ have different colors.
\item $\opcut(v)$: Delete the edge between $v$ and its parent to split
  the tree containing~$v$ into two trees.  This operation assumes
  that~$v$ is not a tree root.
\item $\opevert(v)$: Make~$v$ the root of the tree containing~$v$.
\item $\opfindblue(v)$: Return the pair $(w, x)$ where $x$ is the
  minimum value of a \emph{blue edge} on the tree path from $v$ to
  $\opfindroot(v)$, and $w$ is the \emph{last} vertex on this path
  such that $(w, p(w))$ is a blue edge of cost~$x$.
\item $\opaddblue(v,x)$: Add the real number $x$ to the value of every
  \emph{blue edge} on the tree path from $v$ to
  $\mathsf{findroot}(v)$.
\item $\opaddred(v,x)$: Add the real number $x$ to the value of
  every \emph{red edge} on the tree path from $v$ to $\mathsf{findroot}(v)$.
\end{di}
We explain how to implement this data structure in
Section~\ref{sec:red-blue-link-cut}.

\begin{lemma}
  \label{lem:pruning} Given a matching~$f$ of~$G$ with $s$~edges, we
  can find in time~$O(s \log n)$ a matching~$f_{0}$ with the same
  value such that $H(f_0)$ is a forest and a subgraph of~$H(f)$.
\end{lemma}
\begin{proof}
  We start by creating a red-blue link-cut tree containing isolated
  nodes for the sets~$P$ and~$R$.  $P$~nodes are red, $R$~nodes are
  blue.  We then consider the edges of the matching~$f$ (that is,
  the edges of~$H(f)$) one by one.  Let $(p, r)$ be such an edge. We
  first perform~$\opfindroot(p)$ and~$\opfindroot(r)$ to determine
  if~$p$ and~$r$ are already connected in the current forest.  If not,
  then we merge the two trees using~$\opevert(r)$ and~$\oplink(r, p,
  f(p, r))$.

  If $p$ and~$r$ are already in the same tree, however, then we must
  not create a cycle in the matching graph.  We perform~$\opevert(p)$
  and $\opfindblue(r)$ to retrieve the smallest value~$\Delta$ on any
  blue edge on the path from~$p$ to~$r$ in the current matching.
  Since the edge~$(p, r)$ does not yet exist in the matching, there is
  at least one such blue edge.

  If $f(p, r) \leq \Delta$, then we do not add the
  edge~$(p,r)$. Rather, we perform~$\opaddred(r, f(p, r))$
  and~$\opaddblue(r, -f(p,r))$.  This pushes a flow of value~$f(p,r)$
  along the existing path from~$p$ to~$r$, while keeping the inflow
  and outflow of all intermediate nodes constant.  If $f(p,r) =
  \Delta$, then this could cause the flow on some blue edges to become
  zero.  We identify these edges using~$\opfindblue(r)$, and cut them
  from the tree using~$\opcut$.

  Otherwise, that is if $f(p,r)> \Delta$, we perform~$\opaddred(r,
  \Delta)$ and~$\opaddblue(r, -\Delta)$.  This pushes a flow of
  value~$\Delta$ along the existing path from~$p$ to~$r$, while
  keeping the inflow and outflow of all intermediate nodes
  constant. It also causes at least one blue edge to have value zero.
  We identify all such zero-value blue edges using~$\opfindblue(r)$
  and cut them from the tree using~$\opcut$.  As a result, $p$ and~$r$
  are no longer in the same tree, so we can add the missing flow $f(p,
  r) - \Delta$ by adding the edge~$(p, r)$, using the
  operations~$\opevert(r)$ and~$\oplink(r, p, f(p, r) - \Delta)$.

  There are $s$~edges in the matching~$f$, and handling each edge
  takes a constant number of link-cut tree operations that take
  time~$O(\log n)$, plus possibly the removal of all blue edges whose
  value became zero.  All such edges are cut from the forest, so the
  total number of these cuts is bounded by the number of edges ever
  added, and therefore at most~$s$.
\end{proof}

\subsection{Putting everything together}

\begin{theorem}
  \label{thm:combinatorial}
  Consider a geometric matching problem for a set $P$ of at most $n$
  points and a set $R$ of at most $n$ ranges. Assume that we have a
  compact representation of the incidence graph $I(P,R)$ of size
  $\sigma$.  Then we can compute a maximum matching in $O(n \sigma
  \log \sigma)$ time.
\end{theorem}

We can combine this with the compact representations presented in 
Theorems~\ref{thm:point+rectangle},~\ref{thm:point+disk1} and~\ref{thm:point+halfspace}
to obtain the following:
\begin{itemize}
\item The geometric matching problem for a set $P$
  of at most $n$ points and a set $R$ of at most
  $n$ axis-parallel boxes in~$\Reals^d$, for
  constant~$d$, can be solved in
  $O(n^2 \log^{d+1} n)$ time.
\item The geometric matching problem for a set $P$ of at most $n$
  points and a set $D$ of at most $n$ congruent disks in the plane can
  be solved in $O(n^{7/3} \log^2 n)$ time.
\item The geometric matching problem for a set $P$ of at most $n$
  points and a set $H$ of at most $n$ halfspaces in~$\Reals^d$, for
  constant~$d$, can be solved in $O(n^{\frac{3d+1}{d+1}} \polylog n)$
  time.
\end{itemize}

For the bottleneck versions in the plane, we can solve the decision problem 
for the $L_\infty$-metric using the first item, while the second item 
can be used for the $L_2$-metric.
One can then perform a binary search using the techniques described in 
Section~\ref{sec:bottleneck}, which only add an extra
logarithmic factor to the running time needed to solve the decision problem.


\section{The red-blue link-cut tree}
\label{sec:red-blue-link-cut}

It remains to explain how to implement the red-blue link-cut tree.  We
assume some familiarity with the implementation of link-cut trees, a
data structure that maintains a forest of rooted trees, and which was
introduced by Sleator and Tarjan~\cite{sleator} in order to speed up
Dinitz' algorithm. It has since found many other applications.  Modern
presentations use splay trees~\cite{splay} as the underlying data
structure, see for instance Tarjan and Werneck~\cite{tarjan-werneck}.
Werneck's thesis~\cite{werneck} presents several link-cut tree data
structures in detail.

The implementation of link-cut trees differentiates between two types
of tree edges, \emph{solid} and \emph{dashed}.  Every node~$v$ has at
most one solid edge connecting it to one of its children, and the
other children are connected to~$v$ by dashed edges.  (It is possible
that all children of~$v$ are connected to~$v$ by dashed edges.)  A
\emph{solid path} is a maximal contiguous sequence of solid edges.  A
rooted tree is represented as a set of solid paths.  A specific method
for labeling edges as solid and dashed is essential for the efficiency
of link-cut trees, and the reader is referred
to~\cite{sleator,werneck} for further details.  The following
operation on the solid paths is needed for implementing the list of
link-cut tree operations:
\begin{di}
\item $\opexpose(v)$: Make the tree path from~$v$ to $\opfindroot(v)$
  solid by converting dashed edges along the path to solid, and solid
  edges incident to this path to dashed.
\end{di}

A solid path is represented by a splay tree~\cite{splay}, which is a
binary search tree that can support queries and updates efficiently.
We call this splay tree a \emph{solid tree} to distinguish it from the
represented trees maintained by the link-cut tree data structure.

Since we store values at edges, and we also need the~$\opevert$
operation, we store both the nodes and edges of a solid path as
separate vertices in the solid tree.
So the in-order sequence of the solid tree is an alternating
sequence of vertices representing tree nodes and vertices representing
tree edges, in the order in which they appear on the solid path. (This
is explained in Section~2.1.4 of Werneck's thesis~\cite{werneck}.)

For the implementation of~$\opfindroot$, $\opexpose$, $\oplink$,
and~$\opcut$, we again refer the reader to the
literature~\cite{sleator,werneck}.

The~$\opevert(v)$ operation is implemented by first performing
an~$\opexpose(v)$.  This makes the path from~$v$ to the current root
of the tree a solid path, represented by a splay tree.  We then
reverse the sequence of nodes and edges on this solid path by
``flipping'' the entire splay tree.  To perform this flip operation
efficiently, every vertex of the solid tree conceptually carries a
$\bitreverse$-bit. If this bit is set, then the meaning of left child
and right child is reversed.  What is actually stored in each vertex,
however, is the xor of the vertex's $\bitreverse$-bit with the
$\bitreverse$-bit of its parent in the splay tree (solid tree).  This
allows us to flip an entire splay tree by modifying only the bit
stored at the root of the splay tree.

It remains to discuss the implementation of the $\opfindblue$,
$\opaddblue$, and $\opaddred$ operations.  This is done similar
to~$\opfindcost$ and~$\opaddcost$ in normal link-cut trees, as
follows:

For each vertex~$v$ of the splay tree (solid tree) that represents an
edge, we define
\begin{align*}
  \bitblue(v) &= \text{value of the edge if $v$ represents a blue
    edge, else $\infty$} \\
  \bitred(v) &= \text{value of the edge if $v$ represents a red edge,
    else~$\infty$} \\ 
  \bitminblue(v) &= \text{minimum of $\bitblue(u)$ over all
    vertices~$u$ in the subtree of~$v$ of the solid tree} \\
  \bitminred(v) &= \text{minimum of $\bitred(u)$ over all
    vertices~$u$ in the subtree of~$v$ of the solid tree}
\end{align*}
We do not store this information directly, however, but rather use
four fields in each vertex~$v$ of the splay tree, denoted $\bitDblue$,
$\bitDred$, $\bitDminblue$, and $\bitDminred$.  The first two fields
store
\begin{align*}
  \bitDblue(v) &= \bitblue(v) - \bitminblue(v)\\
  \bitDred(v) &= \bitred(v) - \bitminred(v)
\end{align*}
When $v$ is the root of the solid tree (splay tree),  
the other two fields store
\begin{align*}
  \bitDminblue(v) &= \bitminblue(v) \\
  \bitDminred(v) &= \bitminred(v)
\end{align*}
When $v$ is not the root of the solid tree, then they store
\begin{align*}
  \bitDminblue(v) &= \bitminblue(v) - \bitminblue(\bitparent(v)) \\
  \bitDminred(v) &= \bitminred(v) - \bitminred(\bitparent(v))
\end{align*}
Here, $\bitparent(v)$ denotes the parent of vertex~$v$ in the solid
tree (splay tree) and is not to be confused with the parent~$p(v)$ in
the represented tree.

As an additional complication, when~$\bitreverse(v)$ is true, then the
meaning of red and blue is inversed.  Since the color of an edge is
not a static property but depends on the color of the successor vertex
in the solid tree, this allows us to ``flip'' an entire solid path by
just changing the~$\bitreverse$ bit in the root of the solid tree.

The four fields can be maintained during updates and balancing
operations on the splay tree as in normal link-cut
trees~\cite{sleator,werneck}.

To compute~$\bitblue(v)$ and~$\bitminblue(v)$ for a given vertex~$v$
of the solid tree, we need to follow a path from the root of the solid
tree to the vertex~$v$, maintaining~$\bitreverse(u)$,
$\bitminblue(u)$, and~$\bitminred(u)$ while we follow the path.  We
then have ~$\bitreverse(v)$, $\bitminblue(v)$, and~$\bitminred(v)$,
and can read off the actual value of $\bitminblue(v)$ from the
computed values.  The value $\bitblue(v)$ is then computed as $\bitminblue(v) +
\bitDblue(v)$.

To implement~$\opfindblue(v)$, we first invoke $\opexpose(v)$.  The
path to be considered by the operation is now a solid path.  We can
read off the minimum blue cost from the fields~$\bitDminblue$
and~$\bitDminred$ in the root of the solid tree (depending on
the~$\bitreverse$ bit in the same vertex).  To find the endpoint of
the minimum-cost blue edge, we search in the solid tree.

Finally, to implement~$\opaddred(v, x)$ and~$\opaddblue(v, x)$, we
again invoke~$\opexpose(v)$.  It then suffices to add~$x$ to
either~$\bitDminred$ or~$\bitDminblue$ in the root of the solid tree.


\section{Conclusions}
\label{sec:conclusions}

We have shown in Corollary~\ref{coro:L_1+bottleneck} that computing
the $L_\infty$- and $L_1$-bottleneck distance takes near-linear
time. In contrast, for the $L_2$-bottleneck distance
(Corollary~\ref{coro:L_2+bottleneck}) we have two subroutines that
take roughly $O(n^{4/3})$ time each: the decision problem and the
distance selection.  We conjecture that computing the $L_2$-bottleneck
distance has a lower bound of $\Omega(n^{4/3-\eps})$ for any~$\eps>0$.

Consider the matching problem for points on the real line: we have
two sets~$P$ and~$Q$ of points on the real line, with a total of
$n$~points.  Each point $p\in P$ has a supply $s_p>0$ and each point
$q\in Q$ has a demand~$d_q>0$.  For each $\lambda\ge 0$, let
$G_\lambda(P,Q)$ be the bipartite graph connecting a point of~$P$ to a
point of~$Q$ whenever their distance is at most~$\lambda$.  The
bottleneck value to make the matching problem satisfiable is
\[
\lambda^*(P,Q) ~=~ \min \{ \lambda\ge 0\mid  
G_\lambda(P,Q) \text{ has a satisfying matching}\}.
\]
When the total supply equals the total demand, that is, when
$\sum_{p\in P}s_p = \sum_{q\in Q}d_q$, then the following greedy
matching is optimal: consider the points of~$P$ and of~$Q$ by
increasing $x$-coordinate, and assign each supply to the leftmost
demand that is not satisfied yet.  When total supply and total demand
are not equal, however, then we seem to need again binary search
on~$\lambda$: for a fixed~$\lambda$, we can, again with a greedy
matching, check if a satisfying matching exists.  In both cases,
the total running time is~$O(n \log n)$, but the general case requires
considerably more complicated tools. Is there a simpler algorithm that
does not require binary search on~$\lambda$?

Our discussion has centered on incidence graphs of points and ranges.
However, the same approach can be exploited for \emph{any} bipartite
graph that has a compact representation.  For example, let~$B$ be a set
of \emph{blue} pairwise disjoint segments in the plane and let $R$ be
a set of \emph{red} pairwise disjoint segments in the plane.  Consider
the red-blue intersection graph
\[
X(R,B) = (R\sqcup B, ~ \{ rb\mid r\in R,~ b\in B, ~\text{segments $r$ and $b$ intersect}\}).
\]
It follows from the results of Chazelle et al.~\cite{ChazelleEGS94}
that~$X(R,B)$ has a compact representation of size~$O(n\log^2 n)$ that
can be obtained in $O(n\log^2 n)$~time; see for example~\cite[Section
  4]{MorozA16} for an explicit formulation.  It follows that we can
solve a maximum matching in~$X(R,B)$ in~$O(n^{1+\eps})$ time with high
probability, for any~$\eps>0$.

As another example, consider a simple polygon, and split it with a
diagonal~$d$. Let~$P$ be the set of vertices on one side of~$d$, $Q$
the set of vertices on the other side of~$Q$.  A maximum matching
between~$P$ and~$Q$ such that $p \in P$ and~$q \in Q$ can be matched only
if the are mutually visible inside the polygon can be computed
in~$O(n^{1+\eps})$ time, since the visibility edges admit a
near-linear-size compact representation~\cite{Agarwal1994}.

\subparagraph*{Acknowledgements.}

\OurAcknowledgements

\OurFunding

\bibliography{references}

\end{document}